\documentclass{article}

\usepackage{graphicx}
\usepackage{latexsym}
\usepackage{amsmath}
\usepackage{amsthm}
\usepackage{amssymb}
\usepackage{bm}
\usepackage{tikz}
\usetikzlibrary{arrows}
\usepackage{xcolor}
\tikzstyle{new edge style 1}=[->,line width=2mm,-{Latex[width=0pt 3, length=6mm]}]
\usetikzlibrary{arrows.meta}

\let\shortcite\cite

\newcommand{\rp}{{\rm RP}}
\newcommand{\p}{{\rm P}}
\newcommand{\np}{{\rm NP}}

\newtheorem{mytheorem}{Theorem}

\newtheorem{lemma}[mytheorem]{Lemma}

\newtheorem{example}{Example}

\showhyphens{}

\newcommand{\prob}[3]{
\begin{description}
  \item[Name:] #1
  
  \item[Given:] #2
  
  \item[Question:] #3
\end{description}}

\newcommand{\fsp}{\ensuremath{{f\!s}_p}}
\newcommand{\fvp}{\ensuremath{{fv}_p}}

\begin{document}

\title{Using Weighted Matching to Solve 2-Approval/Veto Control and Bribery}

\author{Zack Fitzsimmons\\
 Dept.\ of Math.\ and Computer Science\\
 College of the Holy Cross\\
 Worcester, MA 01610 \and
  Edith Hemaspaandra\\
  Department of Computer Science\\
  Rochester Institute of Technology \\
  Rochester, NY 14623}%

\date{May 26, 2023}
\maketitle

\begin{abstract}
Determining the complexity of election attack problems is a major research direction in the computational study of voting problems. The paper ``Towards completing the puzzle: complexity of control by replacing, adding, and deleting candidates or voters’’ by Erd{\'e}lyi et al. (JAAMAS 2021) provides a comprehensive study of the complexity of control problems. The sole open problem is constructive control by replacing voters for 2-Approval. We show that this case is in P, strengthening the recent RP (randomized polynomial-time) upper bound due to Fitzsimmons and Hemaspaandra (IJCAI 2022). We show this by transforming 2-Approval CCRV to weighted matching. We also use this approach to show that priced bribery for 2-Veto elections is in P. With this result, and the accompanying (unsurprising) result that priced bribery for 3-Veto elections is NP-complete, this settles the complexity for $k$-Approval and $k$-Veto standard control and bribery cases.
\end{abstract}

\section{Introduction}
Elections are the most widely-used way to aggregate the preferences of a group
of agents over a common set of alternatives. Applications include political
elections as well as multiagent systems in artificial intelligence applications.
Thus it is important to study computational aspects of elections including problems
such as winner determination and different ways of strategically affecting the outcome,
which are referred to as election-attack problems (see, e.g.,~ Brandt et al.~\cite{bra-con-end-lan-pro:b:comsoc-handbook}).

Electoral control and bribery are two important examples of election-attack problems. Control models the actions of an election chair who
has control over the structure of the election (e.g., the set of voters) and modifies this structure to ensure a preferred outcome (e.g., by deleting voters)~\cite{bar-tov-tri:j:control}.
Different control actions model real-world settings such as get-out-the-vote drives or
adding ``spoiler'' candidates to an election to a preferred candidate wins.
The standard control cases of 
adding or deleting voters or candidates can be naturally extended to 
replacing voters or candidates, which for example models settings where the size of the collection of voters is fixed to its initial amount (e.g., in a parliament) and so the control action must work within this restriction~\cite{lor-nar-ros-bre-wal:c:replacing-control}.
Bribery considers how a
strategic agent can set the votes of a subcollection of the voters to
ensure a preferred outcome~\cite{fal-hem-hem:j:bribery}.
The computational study of control and bribery is a major research direction in
computational social choice~(see~Faliszewski and Rothe~\shortcite{fal-rot:b:handbook-comsoc-control-and-bribery}). %

Asking voters to state their most preferred candidate where candidates with the highest
number of approvals win is one of the most natural ways to elicit preferences to reach a group 
decision. This is referred to as the Plurality rule. We
consider $k$-Approval, which generalizes this for fixed $k$. We additionally consider $k$-Veto where
each voter vetoes their $k$ least preferred candidates, and
candidates with the fewest vetoes win.

The study of electoral control has led to many different papers considering the complexity of different actions for important voting rules. A recent comprehensive study on the complexity of control by Erd{\'e}lyi et al.~\shortcite{erd-nev-reg-rot-yan-zor:j:towards-completing} sought to
settle the complexity of the open problems that have remained for control by replacing, adding, and deleting candidates or voters. In the large number of cases summarized and completed
in this work only the complexity of control by replacing voters for 2-Approval remained open. In
very recent work by Fitzsimmons and Hemaspaandra~\shortcite{fit-hem:c:random-classification} this problem
was shown to be in RP (randomized polynomial-time). We strengthen this result and show this problem in \p, thus completing the goal of the comprehensive work by Erd{\'e}lyi et al.~\shortcite{erd-nev-reg-rot-yan-zor:j:towards-completing}.

The complexity of standard control and bribery for $k$-Approval and $k$-Veto elections was systematically
studied by 
Lin~\shortcite{lin:c:manip-k-app,lin:thesis:elections}, which left
only a few open cases. Several of these open cases concern weighted voting and
were handled by Faliszewski, Hemaspaandra, and Hemaspaandra~\shortcite{fal-hem-hem:j:weighted-control}.
The last remaining open cases were 2-Approval and 2/3-Veto priced bribery, where setting the votes of different voters may incur different costs. Bredereck et al.~\shortcite{bre-gol-woe:t:2-approval} state that the 2-Approval case is in \p\ in the summary for a Dagstuhl working group.
The 2-Veto case was incorrectly claimed \np-complete by Bredereck and Talmon~\shortcite{bre-tal:j:edge-cover} and their 
classification of the 3-Veto case relied on the same incorrect proof. We resolve these issues in the
literature in the present paper.

We summarize the main contributions of our paper below.
\begin{itemize}
\item We show that (constructive) 
control by replacing voters for 2-Approval, the sole remaining open case from Erd{\'e}lyi et al.~\shortcite{erd-nev-reg-rot-yan-zor:j:towards-completing}, that was very recently shown to be in \rp\ (randomized polynomial time) by Fitzsimmons and Hemaspaandra~\shortcite{fit-hem:c:random-classification}, is in fact in \p~(Theorem~\ref{thm:2-app-ccrv}).
We mention here that this result was previously reported on in our technical report~\cite{fit-hem:t:random-classification}.

\item We also settle the complexity of the last remaining cases from the study of $k$-Approval and $k$-Veto by Lin~\shortcite{lin:c:manip-k-app,lin:thesis:elections} by showing that priced bribery for 2-Veto is in \p\ (Theorem~\ref{thm:2-veto-pb}), while priced bribery for 3-Veto is \np-complete (Theorem~\ref{thm:3-veto-pb}).

\item To prove our polynomial-time results we use transformations from voting problems to weighted
matching. Weighted matching has been seldom used in the computational study of voting problems and our nontrivial polynomial-time results illustrate its usefulness.
\end{itemize}

\section{Preliminaries}

An election $(C,V)$ consists of a set of candidates $C$ and a collection of voters $V$ where each voter has a corresponding strict total order ranking of the candidates in $C$, referred to as a vote. A voting rule is a mapping from an election
to a subset of the candidate set referred to as the winners.

Our results concern the %
$k$-Approval and $k$-Veto voting rules. In a $k$-Approval election each candidate receives one point from each voter who ranks them among their top-$k$ 
preferred candidates and candidates with the highest score win. In a $k$-Veto election each candidate receives one veto from each voter who ranks them among their bottom-$k$ preferred candidates and candidates with the fewest vetoes win.

\subsection{Election-Attack Problems}

Electoral control and bribery denote well-studied families of election-attack problems (see~Faliszewski and Rothe~\shortcite{fal-rot:b:handbook-comsoc-control-and-bribery}).
To solve important open questions in the literature we consider %
control by replacing voters and
priced bribery.

Control by replacing voters~\cite{lor-nar-ros-bre-wal:c:replacing-control} models scenarios where the election chair with control over the structure of the election replaces a subcollection of the voters with the same number of unregistered voters to ensure a preferred candidate wins. This model ensures that the size of the collection of voters in the election is the same after the control action. %
We formally define this problem below.

\prob{${\cal R}$-Constructive Control by Replacing Voters~(CCRV)~\cite{lor-nar-ros-bre-wal:c:replacing-control}}%
{Given an election $(C,V)$, a collection of unregistered voters $W$, preferred candidate $p$, and integer $k$ such that $0 \le k \le \|V\|$.}%
{Do there exist subcollections $V' \subseteq V$ and $W' \subseteq W$
such that $\|V'\| = \|W'\| \le k$ and $p$ is a winner of the election $(C, (V - V') \cup W')$ using voting rule ${\cal R}$?}

Bribery, introduced by Faliszewski, Hemaspaandra, and Hemaspaandra~\shortcite{fal-hem-hem:j:bribery},
models the actions of an agent who sets the votes of a subcollection of the voters to ensure a preferred candidate wins~\cite{fal-hem-hem:j:bribery}. We mention that bribery can be thought of as modeling the campaign
costs of an election where an agent must determine if there is a subcollection of voters that can be convinced to change their vote to ensure a preferred candidate wins. We consider the problem of priced bribery, which we formally define below.

\prob{${\cal R}$-\$Bribery~\cite{fal-hem-hem:j:bribery}}%
{Given an election $(C,V)$ where each voter $v \in V$ has (integer) price $\pi(v) \ge 0$,
preferred candidate $p$, and budget $k \ge 0$.}%
{Is there a way to change the votes of a subcollection of voters $V' \subseteq V$ within the budget (i.e., $\pi(V') = \sum_{v \in V'} \pi(v) \le k$)
such that $p$ is a winner using voting rule~${\cal R}$?}

As is standard we look at the nonunique-winner model. Erd{\'e}lyi et al.~\cite{erd-nev-reg-rot-yan-zor:j:towards-completing} remark that their proofs can be modified to work for the unique-winner model, and we remark that our proofs can be easily modified to work for the unique-winner model as well.

\subsection{Matching}

For our polynomial-time algorithms, we will reduce to the following polynomial-time computable weighted matching problems (we will be writing $V_G$ for sets of vertices to distinguish them from collections $V$ of voters).

\prob{Max-Weight Perfect $b$-Matching for Multigraphs}%
{An edge-weighted multigraph $G = (V_G,E)$,
where each edge has nonnegative integer weight,
a function $b: V_G \to \mathbb{N}$, and
integer $k \geq 0$.}%
{Does $G$ have a perfect matching of weight at least $k$, i.e., does there exist a set of edges $E' \subseteq E$ of weight at least $k$ such
that each vertex $v \in V_G$ is incident to exactly $b(v)$ edges in $E'$?}

We define Min-Weight Perfect $b$-Matching for Multigraphs analogously.

For simple graphs, Max-Weight and Min-Weight Perfect $b$-Matching are in P~\cite{edm-joh:c:matching}, and  it is easy to reduce Max[Min]-Weight Perfect $b$-Matching for Multigraphs to Max[Min]-Weight Perfect Matching (Lemma~\ref{t:XbtoX}) by generalizing the reduction from Perfect $b$-Matching for Multigraphs to Perfect Matching using the construction from Tutte~\shortcite{tut:j:factor} (see, e.g., Berge~\cite[Chapter 8]{ber:b:graphsbook}).
Lemma~\ref{t:XbtoX} shows this for Max-Weight. The case for Min-Weight is analogous.

\begin{lemma}
\label{t:XbtoX}
Max-Weight Perfect $b$-Matching for Multigraphs reduces to 
Max-Weight Perfect Matching.
\end{lemma}

\def\XbToXproof{%
\begin{proof}
It is straightforward to reduce 
Perfect $b$-Matching for Multigraphs to Perfect Matching
using the construction from Tutte~\shortcite{tut:j:factor} (see, for example, Berge~\cite[Chapter 8]{ber:b:graphsbook}). 
Given multigraph $G$ and capacity function $b$, replace each vertex $v$ by a complete bipartite graph $(P_v,\{v_1, \ldots, v_{\delta(v)}\})$, where $P_v$ is a set of $\delta(v) - b(v)$ padding vertices (we assume wlog that $\delta(v) \geq b(v)$, otherwise there is no matching). 
We number the edges incident with $v$ and for each edge $e$ in $G$, if this edge is the $i$th $v$ edge and the $j$th $w$ edge in $G$, we have an edge $(v_i,w_j)$.

The same reduction works for the weighted version. We set the weight of $(v_i,w_j)$ to the weight of $e$ and we set the weight of the padding edges to 0. Call the resulting graph $G'$. 
It is easy to see that a perfect $b$-matching of $G$ of weight $k$ will give a perfect matching of $G'$ with weight $k$: Simply take all the edges corresponding to the matching of $G$ and add padding edges to make this a perfect matching. Since padding edges have weight 0, the obtained perfect matching of $G'$ has weight $k$. For the converse, note that any perfect matching of $G'$ consists of padding edges plus a set of edges corresponding to a perfect $b$-matching of $G$.~\end{proof}}
\XbToXproof

From the above lemma we have the following.

\begin{mytheorem}
\label{thm:mwpbmm-in-p}
Max-Weight and Min-Weight Perfect $b$-Matching for Multigraphs are in \p.
\end{mytheorem}

Though many matching problems are in $\p$ (see, e.g.,~\cite{ger:b:matching}, Section 7), minor-seeming variations can make these problems hard. 
In the context of voting problems, the weighted version of control by adding voters for 2-Approval
where we limit the number of voters to add
can be viewed as a variation of a weighted matching problem~\cite{lin:c:manip-k-app}. Since weighted control by adding voters  for 2-Approval is $\np$-complete~\cite{fal-hem-hem:j:weighted-control}, so is the corresponding matching problem.
We also mention here that Exact Perfect Matching (which asks if a graph where a subset of the edges is colored red has a perfect matching with exactly a given number of red edges) introduced by Papadimitriou and Yannakakis~\shortcite{pap-yan:j:spanning-tree} and shown to be in \rp\ by Mulmuley, Vazirani, and Vazirani~\shortcite{mul-vaz-vaz:j:matching} is not known to be in~\p.

\section{2-Approval-CCRV is in P}

In the paper 
``Towards Completing the Puzzle: Complexity of Control by Replacing, Adding, and Deleting Candidates or Voters''~\cite{erd-nev-reg-rot-yan-zor:j:towards-completing}, the sole open problem is the complexity of 
constructive control by replacing voters for 2-Approval. 
Very recently, this problem was shown to be in RP (randomized polynomial time)~\cite{fit-hem:c:random-classification}. We now strengthen this upper bound to show the problem in~\p.

\begin{mytheorem}\label{thm:2-app-ccrv}
2-Approval-CCRV is in \p.
\end{mytheorem}

\begin{proof}

We will reduce constructive control by replacing voters for 2-Approval to Max-Weight Perfect $b$-Matching for Multigraphs, which is defined in the preliminaries.

Since Max-Weight Perfect $b$-Matching for Multigraphs is in \p\ (Theorem~\ref{thm:mwpbmm-in-p}), this immediately implies that control by replacing voters for 2-Approval is in \p.

Consider an instance of 2-Approval-CCRV with %
candidates $C$, registered voters $V$, unregistered voters $W$, preferred candidate $p$, and an integer $k \ge 0$. Let $n$ be the size of $V$.
We ask if it is possible to make $p$ a winner by replacing at most $k$ voters from $V$ by a subcollection of $W$ of the same size.
We rephrase our question as follows: Does there exist a subcollection $\widehat{V}$ of $V \cup W$ of size $n$ such that $p$ is a winner and such that $\widehat{V}$ contains at least $n - k$ voters from~$V$?

Note that we can assume that we will put as many voters approving $p$ in $\widehat{V}$ 
as possible. Let $V_p$ be the subcollection of $V$ approving $p$ and let $W_p$ be the subcollection of $W$ approving $p$. 
We put all of $V_p$ in $\widehat{V}$ and as many voters of $W_p$ as possible, keeping in mind that $\widehat{V}$ should contain at most $k$ voters from $W$ and exactly $n$ voters total. This fixes $\fsp$, the final score of $p$.\footnote{We are not claiming that $\fsp$ is the final score of $p$ for every successful control action; however, if there is a successful control action, then there is a successful control action in which the final score of $p$ is $\fsp$ and so we can restrict ourselves to control actions where the final score of $p$ is $\fsp$.}
\[\fsp = \|V_p\| + \min(k, n - \|V_p\|, \|W_p\|)\]
And we re-rephrase our question as follows:
Does there exist a subcollection $\widehat{V}$ of $V \cup W$ such that
\begin{itemize}
\item for each candidate $c$, the score of $c$ is at most $\fsp$,
\item $\widehat{V}$ contains $n$ voters,
\item $\widehat{V}$  contains at least $n - k$ voters from $V$,
\item the score of $p$ is $\fsp$.
\end{itemize}

\begin{figure}

\begin{center}
\begin{tikzpicture}[scale=0.75]
\tikzstyle{new style 0}=[fill=white, draw=black, shape=circle, draw=black]

\tikzstyle{new edge style 0}=[-, dashed]

\pgfdeclarelayer{edgelayer}
\pgfdeclarelayer{nodelayer}
\pgfsetlayers{edgelayer,nodelayer,main}
	\begin{pgfonlayer}{nodelayer}
		\node [style=new style 0] (0) at (-27, 24) {$p$};
		\node [style=new style 0] (1) at (-24, 24) {$b$};
		\node [style=new style 0] (2) at (-20, 24) {$c$};
		\node [style=new style 0] (3) at (-23, 21) {$x$};
		\node [style=new style 0] (4) at (-21.25, 27) {$a$};
		\node (5) at (-23.35, 20.50) {2};
		\node (6) at (-27.25, 24.60) {3};
		\node (7) at (-19.75, 24.60) {3};
		\node (8) at (-21.65, 27.5) {3};
		\node (9) at (-24.25, 24.60) {3};
		\node (10) at (-18.5, 24) {};
		\node (11) at (-16.25, 24) {};
		\node [style=new style 0] (12) at (-15, 24) {$p$};
		\node [style=new style 0] (13) at (-12, 24) {$b$};
		\node [style=new style 0] (14) at (-8, 24) {$c$};
		\node [style=new style 0] (15) at (-11, 21) {$x$};
		\node [style=new style 0] (16) at (-9.25, 27) {$a$};
		\node (17) at (-11.35, 20.50) {2};
		\node (18) at (-15.25, 24.60) {3};
		\node (19) at (-7.75, 24.60) {3};
		\node (20) at (-9.65, 27.5) {3};
		\node (21) at (-12.25, 24.60) {3};
		\node (22) at (-18.75, 29.5) {};
		\node (23) at (-18.75, 29.5) {$V$ consists of};
		\node (25) at (-18.5, 29) {\textbullet\ $b > c > \cdots$};
		\node (26) at (-18.5, 28.5) {\textbullet\  $a > b > \cdots$};
		\node (27) at (-18.5, 28) {\textbullet\  $a > b > \cdots$};
		\node (28) at (-18.5, 27.5) {\textbullet\  $\bm{\mathrel{p} \,> {b} \mathrel{>}\mathrel{\cdots}}$};
		\node (29) at (-18.5, 27) {\textbullet\  $\bm{\mathrel{p} \,> {b} \mathrel{>}\mathrel{\cdots}}$};
		\node (32) at (-15.1, 29.5) {};
		\node (33) at (-15.1, 29.5) {$W$ consists of};
		\node (34) at (-14.85, 29) {\textbullet\  $\bm{\mathrel{a} \,> {c} \mathrel{>}\mathrel{\cdots}}$};
		\node (35) at (-14.85, 28.5) {\textbullet\  $\bm{\mathrel{a} \,> {c} \mathrel{>}\mathrel{\cdots}}$};
		\node (36) at (-14.85, 28) {\textbullet\  $b > c > \cdots$};
		\node (37) at (-14.85, 27.5) {\textbullet\  $\bm{\mathrel{p} \,> {a} \mathrel{>}\mathrel{\cdots}}$};
	\end{pgfonlayer}
	\begin{pgfonlayer}{edgelayer}
		\draw [style=new edge style 0] (4) to (0);
		\draw [bend left=15] (0) to (1);
		\draw [bend right=15] (0) to (1);
		\draw [bend right=15] (4) to (1);
		\draw (4) to (1);
		\draw [style=new edge style 0, bend left=15, looseness=0.75] (4) to (3);
		\draw [style=new edge style 0, bend right=15, looseness=0.75] (4) to (3);
		\draw [style=new edge style 0] (4) to (3);
		\draw [style=new edge style 0, bend left=15] (4) to (2);
		\draw [style=new edge style 0] (4) to (2);
		\draw [style=new edge style 0, bend left=15, looseness=0.75] (1) to (3);
		\draw [style=new edge style 0, bend right=15] (1) to (3);
		\draw [style=new edge style 0] (1) to (3);
		\draw [bend left=15] (1) to (2);
		\draw [style=new edge style 0, bend right=15] (1) to (2);
		\draw [style=new edge style 0, bend left=15] (2) to (3);
		\draw [style=new edge style 0, bend right=15] (2) to (3);
		\draw [style=new edge style 0] (2) to (3);
		\draw [style=new edge style 1] (10.center) to (11.center);
		\draw [style=new edge style 0] (16) to (12);
		\draw [bend left=15] (12) to (13);
		\draw [bend right=15] (12) to (13);
		\draw [style=new edge style 0, bend left=15] (16) to (14);
		\draw [style=new edge style 0] (16) to (14);
		\draw [style=new edge style 0] (13) to (15);
		\draw [style=new edge style 0] (14) to (15);
	\end{pgfonlayer}
\end{tikzpicture}
\end{center}

\caption{Example of the construction from the proof of Theorem~\ref{thm:2-app-ccrv} applied to the instance of 2-Approval-CCRV from Example~\ref{ex:2-app-ccrv} (left) and the perfect $b$-matching of weight $\ge n - k = 2$ corresponding to $\widehat{V}$ (right). The collections of voters $V$ and $W$ are given with the voters in $\widehat{V}$ in bold. In the graph solid edges have weight 1, dashed edges have weight 0, and each vertex is labeled with its $b$-value.}
\label{fig:ex-app}
\end{figure}
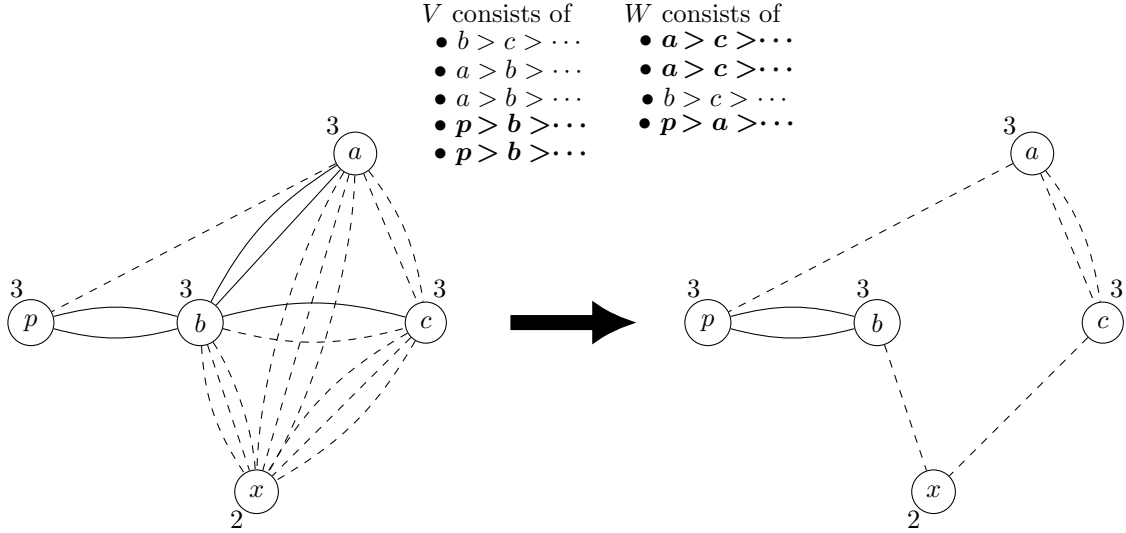

\begin{example} \label{ex:2-app-ccrv} %
Consider the instance of 2-Approval-CCRV with candidates
$\{a,b,c,p\}$,
preferred candidate $p$, $k = 3$, and registered and unregistered voters $V$ and $W$, respectively.

\begin{center}
\begin{minipage}[c]{0.1\textwidth}%
\phantom{.}
\end{minipage}
\begin{minipage}[c]{0.4\textwidth}%
\bigskip
$V$ consists of $n = 5$ voters:
\begin{itemize}
  \item One approving $b$ and $c$.
  \item Two approving $a$ and $b$.
  \item Two approving $p$ and $b$.
\end{itemize}                
\smallskip
\end{minipage}%
\hfill
\begin{minipage}[c]{0.4\textwidth}%
$W$ consists of four voters:
\begin{itemize}
  \item Two approving $a$ and $c$.
  \item One approving $b$ and $c$.
  \item One approving $p$ and $a$.
\end{itemize}
\end{minipage}

\end{center}

\noindent
Initially the scores of $p$, $a$, $b$, and $c$ are 2, 2, 5, and 1, respectively. Note that
\begin{align*}
\fsp &= \|V_p\| + \min(k,n-\|V_p\|, \|W_p\|)\\
&= 2 + \min(3,3,1)\\
&= 3.
\end{align*}

\noindent
$\widehat{V}$ consists of the following $n = 5$ voters.
\begin{itemize}
\item The two voters from $V$ approving $p$ and $b$ (note that $2 \geq n-k$).
\item Three voters from $W$: two approving $a$ and $c$; one approving $p$ and $a$.
\end{itemize}
Notice that for the election after control, with voters $\widehat{V}$, the score of each candidate is at most $\fsp = 3$ and $p$ is a winner with score $\fsp = 3$.
\end{example}

In our reduction,
a voter in $V \cup W$ approving $a$ and $b$ will correspond to an (undirected) edge $(a,b)$. We will call such an edge a ``voter-edge.'' (We will also have padding edges to make everything work.) 

In the graph corresponding to the election, $\widehat{V}$ will correspond to a matching (which will be padded to a perfect matching with padding edges) and the score of a candidate $c$ in $\widehat{V}$ will be the number of edges corresponding to $\widehat{V}$ incident with vertex $c$. We set $b(c)$ to $\fsp$, which will ensure that the number of edges corresponding to $\widehat{V}$ incident with $c$ is at most~$b(c)$.

We still need to ensure the following in our graph.
\begin{itemize}
\item There are exactly $n$ voter-edges in the matching.
\item There are at least $n-k$ voter-edges corresponding to voters in $V$ in the matching.
\item There are exactly $\fsp$ voter-edges in the matching that are incident with $p$ (this ensures that the final score of $p$ is $\fsp$ as advertised).
\end{itemize}

Recall from the discussion in the preliminaries that
for example Exact Perfect Bipartite Matching is not known to be in \p\
(and in fact the weighted version is even \np-complete since Subset Sum straightforwardly reduces to it~\cite{pap-yan:j:spanning-tree}), and so care needs to be taken with handling ``exactness'' restrictions in a matching. This is where we will be using that we are reducing to a \emph{perfect} matching problem.
We will add an extra vertex $x$
that will take up the slack, i.e., it will ensure that the perfect matching contains exactly $n$ voter-edges.
Note that since for every candidate $c$, $b(c) = \fsp$, and we want exactly $n$ voter-edges in the matching, we have a total of $\|C\|\fsp - 2n$ amount of vertex capacity left to cover. 

\medskip
\noindent
{\bf We now formally define our graph $\bm{G}$:}
\begin{description}
\item[Vertices] $V(G) = C \cup \{x\}$.
\begin{itemize}
    \item For each $c \in C$, $b(c) = fs_{p}$.
\item $b(x) = \|C\|\fsp - 2n$. If this is negative, control is not possible.
\end{itemize}
\item [Edges]\phantom{.}
\begin{itemize}
\item For each voter in $V \cup W$ that approves $a$ and $b$, add a voter-edge $(a,b)$.
\item For each candidate $c \in C-\{p\}$, add $b(c) = \fsp$ padding edges $(c, x)$. 
\end{itemize}
\item[Weights]
A perfect matching will saturate $x$ (i.e., contain $b(x) =  \|C\|\fsp - 2n$ edges incident with $x$) and will contain exactly $n$ voter-edges. We want to maximize the number of voter-edges corresponding to voters in $V$, and so we set the weight of those voter-edges to 1 and the weight
of all other edges (i.e., the voter-edges corresponding to $W$ and the padding edges) to 0.
\end{description}

\noindent
{\bf We now prove that our reduction is correct:}
\medskip

\noindent
We will show that $p$ can be made a winner by replacing at most $k$ voters if and only if $G$ has a matching of weight at least $n - k$. See Figure~\ref{fig:ex-app} for an example of this reduction applied to the instance of 2-Approval-CCRV from Example~\ref{ex:2-app-ccrv}.

If $p$ can be made a winner by replacing at most $k$ voters, then there exists
a subcollection $\widehat{V}$ of $V \cup W$ of size $n$ such that $p$ is a winner in $\widehat{V}$, $p$'s score is $\fsp$, and such that $\widehat{V}$ contains at least $n - k$ voters from $V$. Then the edges corresponding to $\widehat{V}$ are a matching of $G$, since the number of edges incident with a candidate $c$ is the score of $c$ in $\widehat{V}$, which is at most $\fsp = b(c)$. And the weight of this matching is at least $n-k$. Since there are $\fsp = b(p)$ edges incident with $p$ in the matching, we can add $b(x)$ edges between $x$ and vertices in $C - \{p\}$ to the matching, to obtain a perfect matching of weight at least $n-k$.

For the converse, suppose $G$ has a perfect matching of weight at least $n - k$. Since the matching is perfect, it contains exactly $n$ voter-edges and at least $n - k$ of these correspond to voters in $V$. Since the matching is perfect, vertex $p$ is incident to exactly $\fsp$ edges in the matching and all of these are voter-edges.
Let $\widehat{V}$ be the collection of  voters corresponding to voter-edges in the matching. It is immediate that $\widehat{V}$ is of size $n$, that
$\widehat{V}$ contains at least $n - k$ voters from $V$, and that $p$ is a winner with score $\fsp$.~\end{proof}

\section{2-Veto-\$Bribery}

We now turn to 2-Veto-\$Bribery and show that this problem is in \p. Lin~\cite[Proof of Theorem~3.8.2]{lin:thesis:elections} showed that Min-Weight $b$-Cover for Multigraphs easily reduces to 2-Veto-\$Bribery,
which implies that %
2-Veto-\$Bribery is unlikely to have a simple polynomial-time algorithm.
In fact, 2-Veto-\$Bribery was incorrectly claimed to be \np-complete by Bredereck and Talmon~\shortcite{bre-tal:j:edge-cover}. The appendix shows why their proof is incorrect.

\begin{mytheorem}
\label{thm:2-veto-pb}
2-Veto-\$Bribery is in \p.
\end{mytheorem}
\begin{proof}
Consider an instance of 2-Veto-\$Bribery with
candidates $C$, voters $V$, preferred candidate $p$, and budget $k \ge 0$. For $v$ a voter, we denote the price of $v$ by $\pi(v)$. If there are at most two candidates the problem is trivial and so we assume we have at least three candidates.
We will be using that Min-Weight Perfect $b$-matching for Multigraphs is in \p~(Theorem~\ref{thm:mwpbmm-in-p}).
In our reduction we will let a voter $v$ who vetoes $a$
and $b$ correspond to an edge $(a,b)$ of weight $\pi(v)$, the price of $v$.
We will call such an edge a voter-edge. We will argue about 2-Veto elections in terms of the number of vetoes each candidate gets (this way each voter affects only two candidates
and these numbers are nonnegative).
Our bribery instance is positive if and only if there is a bribery within budget $k$ such that after bribery the number of vetoes for $p$ is \emph{at most} the number of vetoes for $c$ for every other candidate $c$. 

Let $V_p$ be the collection of voters in $V$ that veto $p$ and let $V'_p$ be the collection of voters in $V$ that do not veto $p$, i.e., $V'_p = V - V_p$. For each $\ell_p, \ell'_p$ such that $0 \leq \ell_p \leq \|V_p\|$ and $0 \leq \ell'_p \leq \|V'_p\|$ we determine whether there exists a successful bribery where $\ell_p$ voters in $V_p$ and $\ell'_p$ voters in $V'_p$ are bribed. Note there are polynomially many such
pairs $\ell_p,\ell'_p$.
Let $\fvp$ be the final number of vetoes for $p$ (i.e., after bribery). Without loss of generality, we assume that we never bribe a voter to veto $p$. So $\fvp = v_p - \ell_p$, where $v_p$ is the number of vetoes for $p$ in $V$.
After bribery, we want each candidate $c \neq p$ to have at least $\fvp$ vetoes. We already mentioned that each voter $v$ that vetoes $a$ and $b$ will correspond to an edge $(a,b)$ of weight $\pi(v)$. We want the collection of voters we bribe to correspond to a matching, so if an edge $(a,b)$ is in the matching, that corresponds to the corresponding voter---who vetoes $a$ and $b$---being bribed, and so the number of vetoes for $a$ and $b$ decreases by 1. 
The obvious thing to do would be to set $b(c)$ to the number of vetoes for $c$ that can be deleted, i.e., $b(c) = v_c - \fvp$, where $v_c$ is the number of vetoes for $c$ in $V$. %
However, things are much more complicated here, since we need to ensure that our matching contains exactly $\ell_p$ voter-edges corresponding to voters in $V_p$ and exactly $\ell'_p$ voter-edges corresponding to voters in $V'_p$, and that the $\ell_p + \ell'_p$ bribed voters each veto two candidates.

At first glance it may seem that we should just simply
bribe only voters that veto $p$ and so $\ell'_p$
would be unnecessary. However,
since our voters have
prices this is not the case. The most obvious
example would be when all voters vetoing $p$ have price
greater than the budget and $p$ is not initially a winner.
Additionally, it is not difficult to construct
scenarios where bribery is possible when voters not
vetoing $p$ are bribed and it is not possible for $p$
to be made a winner by bribing only voters vetoing $p$
(and the prices are each within the budget).
See Example~\ref{ex:2veto-votes} below.

\begin{figure} 
\begin{center}
\begin{tikzpicture}[scale=0.6]
\tikzstyle{new style 0}=[fill=white, draw=black, shape=circle, draw=black, fill=white]
\tikzstyle{graynode}=[draw={rgb,255: red,191; green,191; blue,191}]

\definecolor{mygray}{RGB}{191,191,191}
\tikzstyle{new edge style 0}=[-, dashed]
\tikzstyle{new edge style 2}=[->,line width=2mm,-{Latex[width=0pt 3, length=6mm]}]
\usetikzlibrary{arrows.meta}
\tikzstyle{graydash}=[-, dashed, draw={rgb,255: red,191; green,191; blue,191}]
\tikzstyle{grayedge}=[-, draw={rgb,255: red,191; green,191; blue,191}]
\pgfdeclarelayer{edgelayer}
\pgfdeclarelayer{nodelayer}
\pgfsetlayers{edgelayer,nodelayer,main}
	\begin{pgfonlayer}{nodelayer}
		\node [style=new style 0] (0) at (-27.75, 24) {$a$};
		\node [style=new style 0] (1) at (-24, 24) {$b$};
		\node [style=new style 0] (2) at (-19.25, 24) {$c$};
		\node [style=new style 0] (3) at (-21.5, 22.25) {$x$};
		\node [style=new style 0] (4) at (-20.5, 27) {$p$};
		\node (5) at (-23.85, 19.45) {3};
		\node (6) at (-28, 24.75) {0};
		\node (7) at (-18.55, 23.95) {5};
		\node (8) at (-20.9, 27.7) {2};
		\node (9) at (-24.25, 24.75) {4};
		\node (10) at (-17, 24) {};
		\node (11) at (-14.75, 24) {};
		\node (12) at (-25.5, 24.25) {1};
		\node (13) at (-21.5, 24.6) {1};
		\node (14) at (-21.5, 23.95) {1};
		\node (15) at (-23, 25.7) {1};
		\node (16) at (-22.25, 25.5) {1};
		\node (17) at (-21.0, 25.5) {2};
		\node (18) at (-20.15, 25.5) {2};
		\node (19) at (-19.40, 25.5) {2};
		\node (20) at (-18.60, 25.5) {2};
		\node [style=new style 0] (21) at (-24, 20.25) {$y$};
		\node [style=new style 0] (22) at (-13.5, 24) {$a$};
		\node [style=new style 0] (23) at (-9.75, 24) {$b$};
		\node [style=new style 0] (24) at (-5, 24) {$c$};
		\node [style=new style 0] (25) at (-7.25, 22.25) {$x$};
		\node [style=new style 0] (26) at (-6.25, 27) {$p$};
		\node (27) at (-9.6, 19.45) {3};
		\node (28) at (-13.75, 24.75) {0};
		\node (29) at (-4.30, 23.95) {5};
		\node (30) at (-6.65, 27.7) {2};
		\node (31) at (-10, 24.75) {4};
		\node (32) at (-11.25, 24.25) {\textcolor{mygray}{1}};
		\node (33) at (-7.25, 24.6) {1};
		\node (34) at (-7.25, 23.95) {\textcolor{mygray}{1}};
		\node (35) at (-8.75, 25.7) {1};
		\node (36) at (-8, 25.5) {1};
		\node (37) at (-6.75, 25.5) {\textcolor{mygray}{2}};
		\node (38) at (-5.9, 25.5) {\textcolor{mygray}{2}};
		\node (39) at (-5.15, 25.5) {\textcolor{mygray}{2}};
		\node (40) at (-4.35, 25.5) {\textcolor{mygray}{2}};
		\node [style=new style 0] (41) at (-9.75, 20.25) {$y$};
	\end{pgfonlayer}
	\begin{pgfonlayer}{edgelayer}
		\draw (0) to (1);
		\draw [bend right=15] (4) to (1);
		\draw [bend left=15] (4) to (1);
		\draw [style=new edge style 0, bend left=15, looseness=0.75] (1) to (3);
		\draw [style=new edge style 0, bend right=15] (1) to (3);
		\draw [style=new edge style 0] (1) to (3);
		\draw [bend left=12] (1) to (2);
		\draw [style=new edge style 0, bend left=24] (2) to (3);
		\draw [style=new edge style 0, bend right, looseness=0.75] (2) to (3);
		\draw [style=new edge style 0] (2) to (3);
		\draw [style=new edge style 1] (10.center) to (11.center);
		\draw [bend right=10] (1) to (2);
		\draw [style=new edge style 0, bend right] (1) to (3);
		\draw [style=new edge style 0, bend left=15] (2) to (3);
		\draw [style=new edge style 0, bend right=15, looseness=0.75] (2) to (3);
		\draw [bend right=45] (4) to (2);
		\draw [bend left=75, looseness=1.25] (4) to (2);
		\draw (4) to (2);
		\draw [bend left=45, looseness=0.75] (4) to (2);
		\draw [style=new edge style 0, bend right=15] (1) to (21);
		\draw [style=new edge style 0, bend left=15, looseness=0.75] (1) to (21);
		\draw [style=new edge style 0] (1) to (21);
		\draw [style=new edge style 0, bend right=45, looseness=1.25] (0) to (21);
		\draw [style=new edge style 0, bend right=45] (0) to (21);
		\draw [style=new edge style 0, bend right=45, looseness=1.50] (0) to (21);
		\draw [style=new edge style 0, bend left=40] (2) to (21);
		\draw [style=new edge style 0, bend left=30, looseness=1.25] (2) to (21);
		\draw [style=new edge style 0, bend left=45] (2) to (21);
		\draw [style=grayedge] (22) to (23);
		\draw [bend right=15] (26) to (23);
		\draw [bend left=15] (26) to (23);
		\draw [style=graydash, bend left=15, looseness=0.75] (23) to (25);
		\draw [style=graydash, bend right=15] (23) to (25);
		\draw [style=graydash] (23) to (25);
		\draw [bend left=12] (23) to (24);
		\draw [style=new edge style 0, bend left=24] (24) to (25);
		\draw [style=graydash, bend right, looseness=0.75] (24) to (25);
		\draw [style=graydash] (24) to (25);
		\draw [style=grayedge,bend right=10] (23) to (24);
		\draw [style=new edge style 0, bend right] (23) to (25);
		\draw [style=graydash, bend left=15] (24) to (25);
		\draw [style=graydash, bend right=15, looseness=0.75] (24) to (25);
		\draw [style=grayedge, bend right=45] (26) to (24);
		\draw [style=grayedge,bend left=75, looseness=1.25] (26) to (24);
		\draw [style=grayedge] (26) to (24);
		\draw [style=grayedge,bend left=45, looseness=0.75] (26) to (24);
		\draw [style=graydash, bend right=15] (23) to (41);
		\draw [style=graydash, bend left=15, looseness=0.75] (23) to (41);
		\draw [style=graydash] (23) to (41);
		\draw [style=graydash, bend right=45, looseness=1.25] (22) to (41);
		\draw [style=graydash, bend right=45] (22) to (41);
		\draw [style=graydash, bend right=45, looseness=1.50] (22) to (41);
		\draw [style=new edge style 0, bend left=40] (24) to (41);
		\draw [style=new edge style 0, bend left=30, looseness=1.25] (24) to (41);
		\draw [style=new edge style 0, bend left=45] (24) to (41);
	\end{pgfonlayer}
\end{tikzpicture}
\end{center}
\caption{Example of the construction from the proof of Theorem~\ref{thm:2-veto-pb} applied to the instance of 2-Veto-\$Bribery from Example~\ref{ex:2veto-votes} (left) and the perfect $b$-matching of weight $\le k = 3$ corresponding to $\widehat{V}$ (right). In the graph solid edges are labeled with their corresponding weight, dashed edges have weight 0, and each vertex is labeled with its $b$-value.
Recall that a perfect matching corresponds to the complement of the final bribed election. The ``grayed out'' edges on the right show the final bribed election, which has four voters vetoing $p$ and $c$, one voter vetoing $a$ and $b$, one voter vetoing $b$ and $c$, and three bribe vetoes for $a$ and three bribes vetoes for $b$ (corresponding to the three $(a,y)$ edges and the three $(b,y)$ edges).}
\label{fig:veto-bribery}
\end{figure}
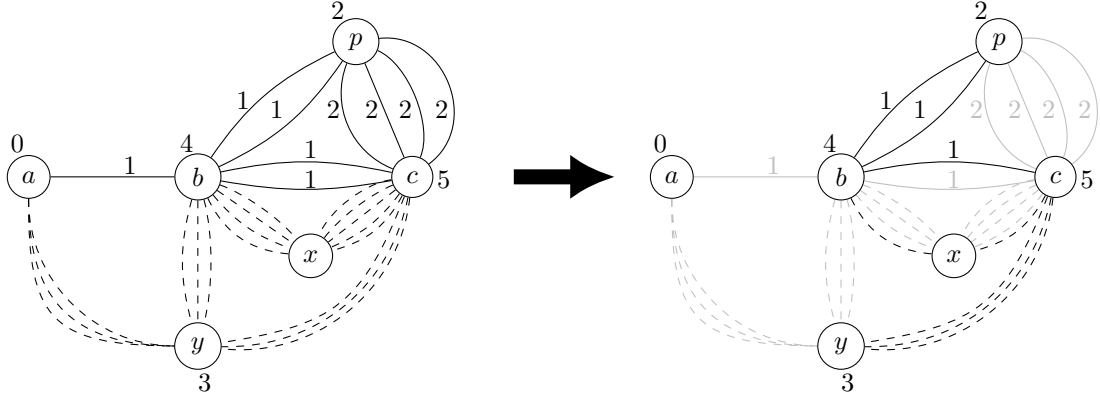

\begin{example}\label{ex:2veto-votes}
Consider the instance of 2-Veto-\$Bribery
with candidates $\{a,b,c,p\}$, preferred candidate $p$, budget $k = 3$, and $V$ consists of the following voters.

\begin{itemize}
  \item One voter vetoing $a$ and $b$ with price 1.
  
  \item Two voters vetoing $b$ and $c$ with price 1.
  
  \item Two voters vetoing $b$ and $p$ with price 1.
  
  \item Four voters vetoing $c$ and $p$ with price 2.
\end{itemize}
From the voters in $V$, $v_a = 1$, $v_b = 5$, $v_c = 6$, and $v_p = 6$.

There is a bribery where $p$ wins where $\ell_p = 2$ and $\ell'_p = 1$ (i.e., that bribes two voters that veto $p$ and one that does not). This sets $\fvp = 6 - 2 = 4$. Bribe the two voters vetoing $b$ and $p$ and one of the voters vetoing $b$ and $c$ to each veto $a$ and $b$. After bribery $a$ and $p$ have 4 vetoes, $b$ and $c$ have 5 vetoes, and so $p$ is a winner.
\end{example}

We rephrase our problem until we are in a matching-friendly form. 

Bribery is possible if and only if there exists a subcollection of voters $\widehat{V} \subseteq V$ (the collection of voters we bribe) such that:
\begin{enumerate}
    \item $\widehat{V}$ contains exactly $\ell_p$ voters of $V_p$ (voters vetoing $p$),
    
    \item $\widehat{V}$ contains exactly $\ell'_p$ voters of $V'_p$ (voters not vetoing~$p$),
    
    \item $\pi(\widehat{V}) \leq k$, and
    
    \item there exists a collection of $\ell_p + \ell'_p$ votes (the votes of the bribed voters after bribery)
    such that for all $c$, the number of vetoes for $c$ after bribery is at least $\fvp$.
\end{enumerate}

Writing $v_c$ for the number of vetoes for $c$ in $V$ and $\widehat{v}_c$ for the number of vetoes for $c$ in $\widehat{V}$, we can rewrite the fourth requirement above as follows.
\begin{enumerate}
\item[4.] There exist integers $bv_c$ for each $c \in C - \{p\}$ (these will be the number of vetoes that the bribery gives to~$c$), such that
\begin{itemize}
\item $v_c - \widehat{v}_c + bv_c\ \text{(the number of vetoes for $c$ after bribery)} \geq \fvp$,
\item $0 \leq bv_c \leq \ell_p + \ell'_p$ (each bribed voter can give only one veto to $c$), and
\item $\sum_{c \in C - \{p\}} bv_c = 2(\ell_p + \ell'_p)$
(the bribery gives $2(\ell_p + \ell'_p)$ vetoes total to %
candidates $C - \{p\}$).
\end{itemize}
\end{enumerate}

Note that the maximum number of vetoes for candidate $c$ after bribery is $v_c + \ell_p + \ell'_p$ (this happens when we do not bribe voters that veto $c$ and all bribed voters are bribed to veto~$c$). If this quantity is less than $\fvp$ for some $c$, bribery is not possible for this choice of $\ell_p$ and $\ell'_p$. Otherwise, $b(c) = v_c + \ell_p + \ell'_p - \fvp$ is the number of vetoes $c$ can lose from the maximum number of possible vetoes for $c$. 

\medskip
\noindent
{\bf We now define our graph $\bm{G}$:}
\medskip

\noindent
Recall that we are reducing to Min-Weight Perfect $b$-Matching for Multigraphs and that we have at least three candidates. 
An example of this construction
applied to the instance of 2-Veto-\$Bribery from Example~\ref{ex:2veto-votes} is presented in Figure~\ref{fig:veto-bribery}.
Note that in order for $G$ to have a $b$-matching, all $b$-values must be nonnegative.

\begin{itemize}
    \item $V(G) = C \cup \{x,y\}$. ($x$ will ensure that the perfect matching contains exactly $\ell'_p$ edges from $V'_p$; this is similar to the role of $x$ in the construction for 2-Approval-CCRV. $y$ will handle the vetoes given by the bribery.)
 \item For each voter in $v \in V$ vetoing $a$ and $b$, add a voter-edge $(a,b)$ with weight $\pi(v)$. 
\item $b(p) = \ell_p$. This will ensure that a perfect matching contains exactly $\ell_p$ voter-edges corresponding to voters that veto $p$.
\item For each $c \in C - \{p\}$, $b(c) = v_c + \ell_p + \ell'_p - \fvp$.
\item  For each $c \in C - \{p\}$, add $(\ell_p + \ell'_p)$ bribe-edges $(c,y)$ with weight 0. This is the maximum number of vetoes that $c$ can receive from the bribery. The bribe-edges \emph{not} in the perfect matching will correspond to the vetoes that the bribery gives to $c$.\footnote{At this point one might wonder why we are having the {\em complement} of the matching correspond to the final election rather than the matching itself. The reason is that we would then be arguing about covers (instead of matchings) which would have to be padded to perfect covers. Padding covers is a little harder than padding matchings; with matchings we just add edges, but with covers we also need to update the $b$-values.} Since the bribery gives a total of $2(\ell_p + \ell'_p)$ vetoes
and we have a total of $\|C-\{p\}\|(\ell_p + \ell'_p)$ bribe-edges, we 
set $$b(y) = (\|C\|-3)(\ell_p + \ell'_p).$$

\item For each $c \in C-\{p\}$, add $b(c)$ padding edges $(c,x)$ of weight 0.
We want to ensure that a perfect matching contains exactly $\ell'_p$ voter-edges corresponding to voters that do not veto $p$. The excess of $b$-values is 
\renewcommand{\labelitemii}{$\circ$}
\begin{itemize}
    \item $\sum_{c \in C - \{p\}} b(c)$ %
    \item minus $\ell_p$ (one for each of the  $\ell_p$ voter-edges incident with $p$; note that each such edge is incident with one $c \in C - \{p\})$
    \item minus $2 \ell'_p$ (two for each of the $\ell'_p$ voter-edges not incident with~$p$)
    \item minus $b(y)$ (the number of bribe-edges in the perfect matching; note that each such edge is incident with one $c \in C - \{p\}$).
    \end{itemize}
We set $b(x)$ to this value.
\end{itemize}

\noindent
{\bf We now prove that our reduction is correct:}
\medskip

\noindent
We need to show that there is a successful bribery that bribes $\ell_p$ voters from $V_p$
and $\ell'_p$ voters from $V'_p$ of cost $\le k$ if and only if $G$ has perfect $b$-matching
of weight $\le k$. Note that $G$ is independent of $k$, and so it suffices to prove the following.

\begin{quote}
For all $k$ there is a successful bribery that bribes $\ell_p$ voters from $V_p$ and 
$\ell'_p$ voters from $V'_p$ of cost $k$ if and only if $G$ has a perfect $b$-matching of weight $k$. 
\end{quote}

If there is a successful bribery of cost $k$
that bribes $\ell_p$ voters from $V_p$ and $\ell'_p$ voters from $V'_p$, let $\widehat{V}$ be the collection of bribed voters, and let $\fvp$, $v_c$, and $\hat{v}_c$ be as defined above.
For $c \in C - \{p\}$, let $bv_c$ be the number of vetoes the bribery gives to $c$. Without loss of generality, assume that the bribery does not give vetoes to $p$. Recall that we have at least three candidates.

We will show that the following set of edges $E'$ is a perfect $b$-matching of $G$ of weight $k$.

\begin{itemize}
\item the voter-edges corresponding to $\widehat{V}$ (since these are the only nonzero weight edges, the matching will have total weight $\pi(\widehat{V}) = k$),
\item for each $c \in C - \{p\}$, $(\ell_p + \ell'_p) - bv_c$
bribe-edges $(c,y)$ representing the vetoes not given to $c$ by the bribery,
\item for each $c \in C - \{p\}$, padding edges $(c,x)$ to meet the corresponding $b$-value. Since $c$ is incident with $\widehat{v}_c$ voter-edges and $(\ell_p + \ell'_p) - bv_c$ bribe-edges, we need $b(c) - \widehat{v}_c - ((\ell_p + \ell'_p) - bv_c)$ padding edges.
\end{itemize}

We now show that $E'$ is a perfect $b$-matching of $G$. Part of showing this is making sure that all the numbers of edges are nonnegative.

\begin{itemize}
\item The edges in $E'$ incident with $p$ are exactly the $\ell_p$ voter-edges in $\widehat{V}$ that veto $p$ and $b(p) = \ell_p \geq 0$.

\item For
$c \in C - \{p\}$, the definition ensures that the number of voter-edges + bribe-edges + padding edges in $E'$ incident with $c$ add up to $b(c)$. But we have to ensure that 
these numbers are all nonnegative. This is immediate for the number of voter-edges incident with $c$ and for the number of bribe-edges incident with $c$, and so it remains to show that the number of padding edges incident with $c$, $b(c) - \hat{v}_c - ((\ell_p + \ell'_p) - bv_c)$, is nonnegative. Since $b(c) = v_c + \ell_p + \ell'_p - \fvp$, we are looking at the quantity $v_c + \ell_p + \ell'_p - \fvp 
- \hat{v}_c - ((\ell_p + \ell'_p) - bv_c) = v_c - \fvp - \hat{v}_c + bv_c$. The number of vetoes for $c$ after bribery is $v_c - \hat{v}_c + bv_c$, and this number is $\geq \fvp$, since $p$ is a winner. It follows that the number of padding edges incident with $c$ in $E'$ is nonnegative.
\item
The number of edges in $E'$ incident with $y$ is
\begin{align*}
\sum_{c \in C - \{p\}} ((\ell_p + \ell'_p) - bv_c) \ &=\sum_{c \in C - \{p\}} (\ell_p + \ell'_p) - \sum_{c \in C - \{p\}}bv_c\\
&=\|C-\{p\}\|(\ell_p + \ell'_p) - 2(\ell_p + \ell'_p)\\
&=b(y).
\end{align*}
This is nonnegative, since we have at least three candidates.

\item
The number of edges in $E'$ incident with $x$ is the number of padding edges in $E'$, which is nonnegative by the argument for $c \in C - \{p\}$ above.
This quantity is
$$\sum_{c \in C - \{p\}} (b(c) - \hat{v}_c - ((\ell_p + \ell'_p) - bv_c)) =
\sum_{c \in C - \{p\}} b(c) -
\sum_{c \in C - \{p\}} \hat{v}_c - 
\sum_{c \in C - \{p\}}((\ell_p + \ell'_p) - bv_c).$$

Since
$$\sum_{c \in C - \{p\}} \hat{v}_c = \ell_p + 2\ell'_p\ \text{ and} 
\sum_{c \in C - \{p\}}((\ell_p + \ell'_p) - bv_c)
= b(y),$$
this is exactly $b(x)$.
\end{itemize}

For the converse, suppose $G$ has a perfect matching of weight $k$.
Let $\widehat{V}$ be the collection of voters corresponding to the voter-edges in the matching.
We will show that we can make $p$ a winner by bribing the voters in $\widehat{V}$. Note that the cost of this bribery is $k$. 
For $c \in C - \{p\}$, let $bv_c$ be the number of bribe-edges incident with $c$ that are not in the perfect matching. We will show that 
\begin{enumerate}
\item If the bribery gives $bv_c$ vetoes to each candidate $c \in C-\{p\}$ and no vetoes to $p$, then $p$ is a winner, and 
\item there are $\ell_p + \ell'_p$ votes that give $bv_c$ vetoes to each candidate $c \in C - \{p\}$ and no vetoes to $p$.
\end{enumerate}
The second item follows immediately from the fact that $\sum_{c \in C -\{p\}} bv_c = 2(\ell_p + \ell'_p)$ (since there are exactly $2(\ell_p + \ell'_p)$ edges incident with $y$ that are not in the matching), and for each $c \in C - \{p\}, bv_c \leq \ell_p + \ell'_p$ (since there are $\ell_p + \ell'_p$ bribe-edges incident with $c$ in $G$).

For the first item, we need to show that for each candidate $c \in C - \{p\}$,
$v_c - \hat{v}_c + bv_c$ (the number of vetoes for $c$ after bribery) is at least $\fvp$.

From the matching, we have that for each candidate $c \in C - \{p\}$,
$\hat{v}_c$ (the number of voter-edges incident with $c$ in the matching)
$+ (\ell_p + \ell'_p - bv_c)$ (the number of bribe-edges incident with $c$ in the matching)
$ \leq b(c) = v_c + \ell_p + \ell'_p - \fvp$.
This implies that $\hat{v}_c - bv_c \leq v_c - \fvp,$
which immediately implies that
$\fvp \le v_c - \hat{v}_c + bv_c,$ and so $p$ is a winner.~\end{proof}

\section{3-Veto-\$Bribery is NP-complete}

It should come as no surprise that 3-Veto-\$Bribery is NP-complete. This was claimed in Bredereck and Talmon~\shortcite{bre-tal:j:edge-cover}, but with an incorrect proof (a simple modification of their incorrect NP-hardness proof for 2-Veto-\$Bribery).

\begin{mytheorem}\label{thm:3-veto-pb}
3-Veto-\$Bribery is \np-complete.
\end{mytheorem}

\begin{proof}
We give a straightforward reduction from Restricted Exact Cover by 3-Sets to our bribery problem.
The well-known \np-complete problem of Exact Cover by 3-Sets (X3C) asks, given a finite collection of 3-element subsets of some finite set $B$, whether the collection has an exact cover, i.e., whether 
there exists a subcollection such that each element of $B$ occurs exactly once~\cite{kar:b:reducibilities}.
We reduce from Restricted Exact Cover by 3-Sets (RX3C), which is the restriction of X3C where each element of $B$ occurs in exactly three subsets. This restricted problem is still NP-complete~\cite{gon:j:restricted-exact-cover}.

\prob{Restricted Exact Cover by 3-Sets (RX3C)}%
{Given a set $B = \{b_1, \dots, b_{3k}\}$ and a collection  ${\cal S} = \{S_1, \dots, S_n\}$ of 3-element subsets of $B$ where each $b \in B$ occurs in exactly three subsets $S_i$ (so $n = 3k$).}%
{Does there exist a subcollection ${\cal S'}$ of ${\cal S}$ that forms an exact cover of $B$?}

Given an instance of RX3C, $B = \{b_1, \dots, b_{3k}\}$ and ${\cal S} = \{S_1, \dots, S_n\}$, we construct an instance of 3-Veto-\$Bribery in the following way.

The set of candidates consists of $B$, $p$, two padding candidates $p_1$ and $p_2$, and $3k$ dummy candidates $d_1, \dots, d_{3k}$. The preferred candidate is $p$ and the budget is $k$.
The voters are defined as follows.
\begin{itemize}
\item For each set $S \in {\cal S}$ with $S = \{x,y,z\}$
we have one voter that vetoes $x$, $y$, and $z$, with price 1.
\item We have two voters vetoing  $p_1$, $p_2$, and $p$, each with price $k+1$.
\item We have $k$ voters, each with price $k+1$, with votes that ensure each of the $d_i$ candidates receive one veto.
\end{itemize}
Note that each dummy candidate has one veto, $p$ and the two padding candidates have two vetoes, and each candidate in $B$ has three vetoes.

If $p$ can be made a winner by bribing a subcollection of voters with total price at most $k$ %
we will show that there is an exact cover ${\cal S}'$ of ${\cal S}$.
Only the voters with votes corresponding to sets in ${\cal S}$ can be bribed,
since all other voters have prices greater than the budget. This means that
$p$ will have at least two vetoes in any bribed election, and a successful bribery would need each of the $3k$ dummy candidates $d_i$ to gain a veto and each $b_i$ candidate can lose at most one veto. Therefore voters  corresponding to sets in ${\cal S}$ will correspond to an exact cover.

For the converse, suppose there exists an exact cover ${\cal S}'$ of ${\cal S}$. Since each $b \in B$ occurs exactly three times in ${\cal S}$ we know $\|{\cal S}'\| = k$.
We can construct a successful bribery by bribing the $k$  voters corresponding to ${\cal S'}$ to veto the
$3k$ dummy candidates $d_i$.
Since ${\cal S'}$ is an exact cover, each $b_i$ candidate will
lose exactly one veto, and so $p$ is a winner.~\end{proof}

\section{Pushing the Boundary}

As we've seen in the previous section, weighted matching turned out to be a useful tool to solve some open control and bribery problems. These problems inherit their complicated algorithms from weighted matching. In this section we discuss how far we can push this technique; can we provide polynomial-time algorithms for more general problems?

It turns out that we can. For the sake of such an example we will show that even after adding
prices to the registered and unregistered voters in control by replacing voters for 2-Approval, this problem will
remain in \p.\footnote{Though it may seem from our results that control problems for 2-Approval will all be in \p\ this is not the case. Control by adding and control by deleting candidates are each \np-complete~\cite{lin:c:manip-k-app,lin:thesis:elections}, and weighted control by deleting voters is \np-complete~\cite{fal-hem-hem:j:weighted-control}.}
Priced control was introduced by Miasko and Faliszewski~\shortcite{mia-fal:j:priced-control} for control by adding and deleting voters and candidates. 
In priced control, each control action has a price and the question is whether it is possible to obtain the desired result by control within a budget.
We can extend this notion to control by replacing voters, by having each registered voter $v \in V$ have a price $\pi(v)$,
each unregistered voter $w \in W$ have price $\pi(w)$,
and letting $k$ be the budget.

We will now explain how to modify the proof of Theorem~\ref{thm:2-app-ccrv} using some of the ideas from the proof of Theorem~\ref{thm:2-veto-pb} to also work for priced control by replacing voters for 2-Approval. %
Unlike in the unpriced case, we do not know what $\fsp$ will be. And so, similarly as in the proof of Theorem~\ref{thm:2-veto-pb}, we will brute-force over all (polynomially many) possible values of $\fsp$. For each value of $\fsp$, $0 \le \fsp \le \|V\|$, we look at the graph defined in the proof of Theorem~\ref{thm:2-app-ccrv}. The only difference will be the weights. All we have to do is to set the weight of each voter-edge that corresponds to a voter $v \in V$ to $\pi(v)$ instead of 1 and each voter-edge that corresponds to a voter $w \in W$ to $-\pi(w)$ instead of 0.
The same argument as in the proof of Theorem~\ref{thm:2-app-ccrv} shows that there is a successful control action within budget $k$ if and only if $G$ has a perfect matching of weight $\pi(V) - k$.

 An immediate corollary of the above is that this implies that 2-Approval-\$Bribery is in \p\ (as Bredereck et al.~\shortcite{bre-gol-woe:t:2-approval} state).
 We can solve 2-Approval-\$Bribery using the result just stated for priced control by replacing voters: %
  The registered voters (including price) will be the $n$ voters from the 2-Approval-\$Bribery instance that we want to solve, the budget will be same,
 and the unregistered voters will consist of ``all possible votes'' that we can bribe to, each with price 0, i.e., for each pair of candidates $a,b$, the collection of unregistered voters contains $n$ voters approving $a$ and $b$, each with price 0.

\section{Conclusion} %

We showed that by using weighted matching we were able to solve the sole open
problem from the comprehensive study of
control by Erd{\'e}lyi et al.~\shortcite{erd-nev-reg-rot-yan-zor:j:towards-completing}.
We also used weighted matching to solve the remaining standard bribery cases for
$k$-Veto. Overall this settles the complexity for $k$-Approval and $k$-Veto for standard control
and bribery cases.
Unlike unweighted matching, weighted matching is not often used
to solve voting problems (though it was, e.g., recently used in a multiwinner election setting~\cite{cel-hua-vis:c:multiwinner-fairness} and for bipartite graphs to compute
the distance between elections~\cite{fal-sko-sli-szu-tal:c:similar-elections}), 
but as our results show, it is a powerful technique whose potential deserves to be further explored.

\section*{Acknowledgements}

This work was supported in part by NSF-DUE-1819546.
We thank the anonymous reviewers for their helpful comments.

\newcommand{\etalchar}[1]{$^{#1}$}

\appendix

\section{Appendix}

\subsection*{Counterexample for 2-Veto-\$Bribery NP-completeness result from~Bredereck and Talmon~[2016]}

As mentioned in the main text, 2-Veto-\$Bribery was incorrectly claimed to be \np-complete by Bredereck and Talmon~\shortcite{bre-tal:j:edge-cover}.
The flaw in their argument is an incorrect \np-hardness reduction for Min-Weight $b$-Cover for Multigraphs (there called Simple Weighted $b$-Edge Cover), which asks given an edge-weighted multigraph $G = (V_G,E)$, function $b: V \to \mathbb{N}$, and integer $k \ge 0$, if there is a set of edges $E' \subseteq E$ of weight at most $k$ such that each vertex is incident to at least $b(v)$ edges. This problem is in fact in P~\cite{edm-joh:c:matching}, which directly implies that their reduction is incorrect (assuming $\p \neq \np$). We provide the counterexample below only to show where their proof fails for the interested reader.
We contacted Bredereck and Talmon, who confirm the bug~\cite{bre:misc:bug}.
Note that in our counterexample we generally use the notation from Bredereck and Talmon~\cite{bre-tal:j:edge-cover} to make it easier to compare directly to their paper.

Their proof claims to show this problem \np-hard with a reduction from the following version of Numerical Matching with Target Sums (NMTS): Given three multisets of positive integers $A = \{a_1, \dots, a_n\}$, $B = \{b_1, \dots, b_n\}$, and $C = \{c_1, \dots, c_n\}$ where all $3n$ integers are distinct, does there exist a partition
of $A \cup B$ into $n$ pairs such that for each $i$, $1 \le i \le n$, the $i$th pair consists of an element from $A$ and an element from $B$ that sum to $c_i$?

The reduction in the proof of Theorem~2 from Bredereck and Talmon claims to
construct a graph from a given NMTS instance such that there is a $b$-Edge Cover of weight at most $8n + 6(\#_3-n)+2n(n-1)n^4$ (where
$\#_3$ denotes that number of triples $a_i$, $b_j$, $c_\ell$ such that $a_i + b_j = c_\ell$) if and only if the NMTS instance is positive.
We now present our counterexample.

Consider the following negative instance of NMTS with $n = 2$.
\begin{itemize}
\item $A = \{3, 4\}$
\item $B = \{5, 6\}$
\item $C = \{8, 9\}$
\end{itemize}

The graph constructed using the reduction in the proof of Theorem~2 from
Bredereck and Talmon~\cite{bre-tal:j:edge-cover} is presented in Figure~\ref{fig:counterexample}.
We will show that this graph in fact has a cover of weight
$22 + 4n^4 = 8n + 6(\#_3-n)+2n(n-1)n^4$, which shows that reduction is incorrect.

\begin{figure}[h!]
\begin{center}
\begin{tikzpicture}[scale=0.775]
\tikzstyle{new style 0}=[fill=white, draw=black, shape=circle, draw=black]
\tikzstyle{none}=[]

\tikzstyle{new edge style 0}=[-, dashed]

\pgfdeclarelayer{edgelayer}
\pgfdeclarelayer{nodelayer}
\pgfsetlayers{edgelayer,nodelayer,main}
	\begin{pgfonlayer}{nodelayer}
		\node [style=new style 0] (0) at (10, 4.85) {row$_2$};
		\node [style=new style 0] (1) at (10, 14.30) {row$_1$};
		\node [style=new style 0] (2) at (14, 10) {col$_2$};
		\node [style=new style 0] (3) at (6.25, 10) {col$_1$};
		\node [style=new style 0] (4) at (7.75, 7.5) {$v(\{4,5\})$};
		\node [style=new style 0] (5) at (12.25, 7.5) {$v(\{4,6\})$};
		\node [style=new style 0] (6) at (7.75, 12.25) {$v(\{3,5\})$};
		\node [style=new style 0] (7) at (12.25, 12.25) {$v(\{3,6\})$};
		\node [style=new style 0] (8) at (1, 11) {$8$};
		\node [style=new style 0] (9) at (4, 11) {$3$};
		\node [style=new style 0] (10) at (2.5, 10.25) {$\star$};
		\node [style=new style 0] (11) at (1, 9.5) {$7$};
		\node [style=new style 0] (12) at (4, 9.5) {$4$};
		\node [style=new style 0] (13) at (1.75, 12) {$1$};
		\node [style=new style 0] (14) at (3.5, 12) {$2$};
		\node [style=new style 0] (15) at (1.5, 8.5) {$6$};
		\node [style=new style 0] (16) at (3.5, 8.5) {$5$};
		\node [style=none] (18) at (2.5, 9.6) {$4$};
		\node [style=new style 0] (19) at (17.75, 8.5) {$8$};
		\node [style=new style 0] (20) at (20.75, 8.5) {$3$};
		\node [style=new style 0] (21) at (19.25, 7.75) {$\star$};
		\node [style=new style 0] (22) at (17.75, 7) {$7$};
		\node [style=new style 0] (23) at (20.75, 7) {$4$};
		\node [style=new style 0] (24) at (18.5, 9.5) {$1$};
		\node [style=new style 0] (25) at (20.25, 9.5) {$2$};
		\node [style=new style 0] (26) at (18.25, 6) {$6$};
		\node [style=new style 0] (27) at (20.25, 6) {$5$};
		\node [style=none] (28) at (19.25, 7.1) {$4$};
		\node [style=new style 0] (29) at (4, 1.75) {$8$};
		\node [style=new style 0] (30) at (7, 1.75) {$3$};
		\node [style=new style 0] (31) at (5.5, 1) {$\star$};
		\node [style=new style 0] (32) at (4, 0.25) {$7$};
		\node [style=new style 0] (33) at (7, 0.25) {$4$};
		\node [style=new style 0] (34) at (4.75, 2.75) {$1$};
		\node [style=new style 0] (35) at (6.5, 2.75) {$2$};
		\node [style=new style 0] (36) at (4.75, -0.5) {$6$};
		\node [style=new style 0] (37) at (6.5, -0.5) {$5$};
		\node [style=none] (38) at (5.55, 0.35) {$4$};
		\node [style=new style 0] (39) at (18.25, 4.25) {$v(9^1)$};
		\node [style=new style 0] (40) at (20.25, 4.25) {$v(9^2)$};
		\node [style=new style 0] (41) at (1.5, 6.5) {$v(8^1)$};
		\node [style=new style 0] (42) at (3.5, 6.5) {$v(8^2)$};
		\node [style=none] (43) at (13.0, 13.35) {$2$};
		\node [style=none] (44) at (7.0, 13.35) {$2$};
		\node [style=none] (45) at (13.3, 6.75) {$2$};
		\node [style=none] (46) at (6.7, 6.75) {$2$};
		\node [style=none] (47) at (6.6, 8.75) {$n^4$};
		\node [style=none] (48) at (13.75, 8.75) {$n^4$};
		\node [style=none] (49) at (13.5, 11.25) {$n^4$};
		\node [style=none] (50) at (11.25, 13.75) {$n^4$};
		\node [style=none] (51) at (8.75, 13.75) {$n^4$};
		\node [style=none] (52) at (6.5, 11.25) {$n^4$};
		\node [style=none] (53) at (11.25, 5.75) {$n^4$};
		\node [style=none] (54) at (8.75, 5.75) {$n^4$};
	\end{pgfonlayer}
	\begin{pgfonlayer}{edgelayer}
		\draw (7) to (2);
		\draw (2) to (5);
		\draw (5) to (0);
		\draw (0) to (4);
		\draw (4) to (3);
		\draw (3) to (6);
		\draw (6) to (1);
		\draw (7) to (1);
		\draw (14) to (10);
		\draw (10) to (9);
		\draw (10) to (12);
		\draw (10) to (16);
		\draw (10) to (15);
		\draw (10) to (11);
		\draw (10) to (8);
		\draw (10) to (13);
		\draw (8) to (11);
		\draw (9) to (12);
		\draw (25) to (21);
		\draw (21) to (20);
		\draw (21) to (23);
		\draw (21) to (27);
		\draw (21) to (26);
		\draw (21) to (22);
		\draw (21) to (19);
		\draw (21) to (24);
		\draw (19) to (22);
		\draw (20) to (23);
		\draw (35) to (31);
		\draw (31) to (30);
		\draw (31) to (33);
		\draw (31) to (37);
		\draw (31) to (36);
		\draw (31) to (32);
		\draw (31) to (29);
		\draw (31) to (34);
		\draw (29) to (32);
		\draw (30) to (33);
		\draw (4) to (34);
		\draw (4) to (35);
		\draw (7) to (24);
		\draw [bend left=15] (7) to (25);
		\draw (39) to (26);
		\draw (27) to (40);
		\draw [bend left=45, looseness=1.25] (39) to (36);
		\draw [bend left] (40) to (37);
		\draw [bend left] (13) to (6);
		\draw (14) to (6);
		\draw (15) to (41);
		\draw (16) to (42);
	\end{pgfonlayer}
\end{tikzpicture}
\end{center}
\caption{Graph constructed using the reduction in the proof of Theorem 2 from Bredereck and Talmon~\cite{bre-tal:j:edge-cover} given the negative NMTS instance $A = \{3,4\}$, $B = \{5,6\}$, $C=\{8,9\}$, and $n = 2$. Vertices and edges have $b$-values and weights of 1 except where labeled explicitly.
We will show that this graph has a cover of weight
$22 + 4n^4$.
} 
\label{fig:counterexample}
\end{figure}
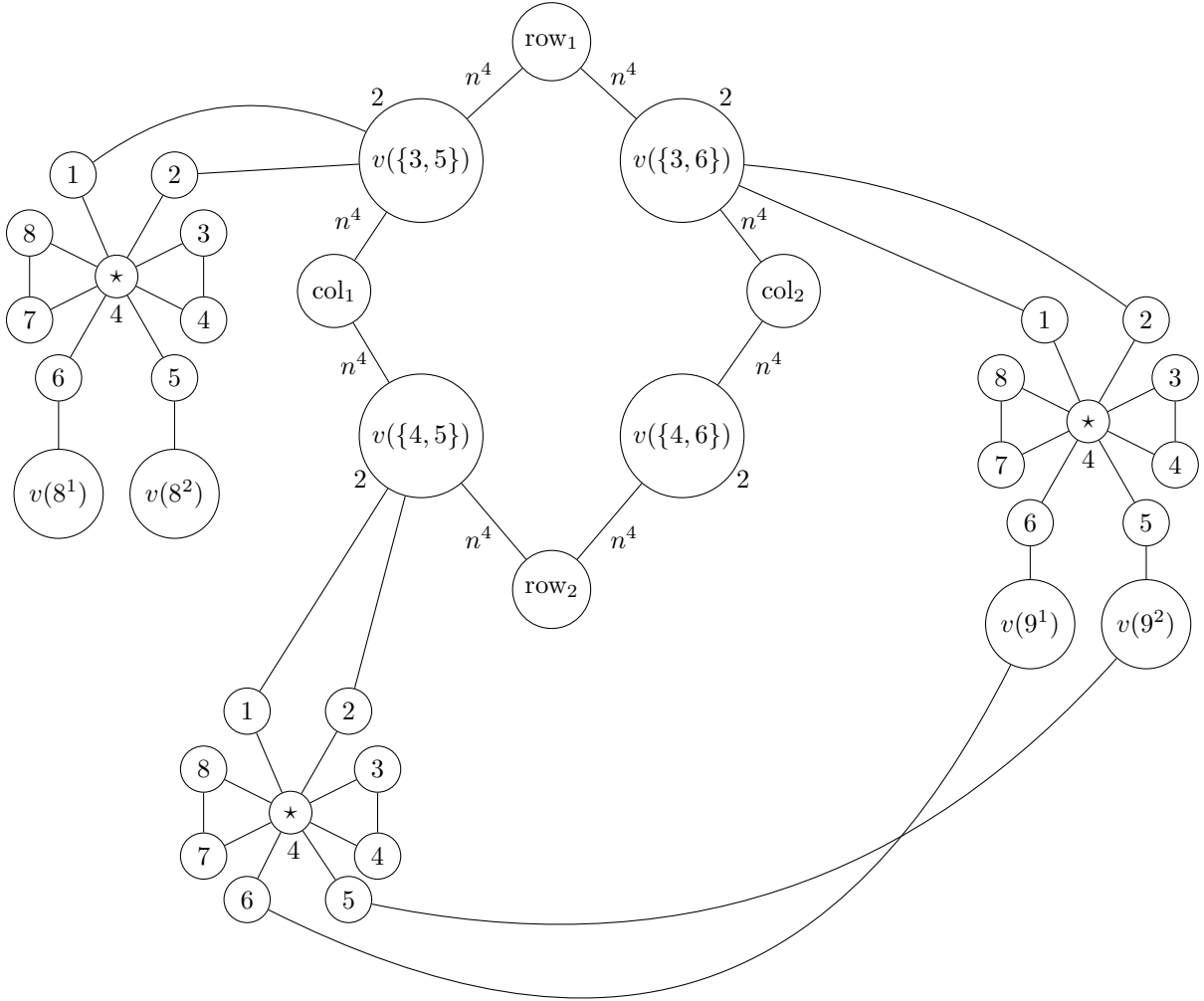

\newpage

We now show that this graph has the following $b$-Edge Cover of weight $22 + 4n^4$.
We describe the cover below (we also illustrate this cover in Figure~\ref{fig:counterexample-cover}).

We first handle the 4 heavy (weight $n^4$) edges.
\begin{itemize}
\item $(\text{row}_1, v(\{3,6\}))$
\item $(\text{row}_2, v(\{4,6\}))$
\item $(\text{col}_1, v(\{4,5\}))$
\item $(\text{col}_2, v(\{4,6\}))$
\end{itemize}
This handles the demand of the vertices row$_1$, row$_2$, col$_1$, col$_2$, and $v(\{4,6\})$. We still have to cover $v(\{3,6\})$ and $v(\{4,5\})$ once and $v(\{3,5\})$ twice.

We can handle the demand of the vertex $v(\{3,5\})$, its corresponding connectivity gadget (the nine-vertex subgraph that connects it to the vertices $v(8^1)$ and $v(8^2)$), and the target vertices $v(8^1)$ and $v(8^2)$ by adding the second configuration described in the proof from Bredereck and Talmon~\cite{bre-tal:j:edge-cover}. This adds 8 weight-1 edges to the cover and meets the demand of all of the involved vertices.

What’s left is to handle the remaining demand (which equals 1) for vertices $v(\{3,6\})$, $v(\{4,5\})$, their corresponding connectivity gadgets, and the target vertices $v(9^1)$ and $v(9^2)$. We handle this by finding a cover of the two gadgets that does not correspond to either of the expected configurations in the proof from Bredereck and Talmon~\cite{bre-tal:j:edge-cover}.

For $v(\{3,6\})$, add the following edges to the cover.
\begin{itemize}
\item $(v(\{3,6\}), 1)$
\item $(2, \star)$
\item $(7, \star)$
\item $(8, \star)$
\item $(5, \star)$
\item $(3, 4)$
\item $(6, v(9^1))$
\end{itemize}
This adds 7 weight-1 edges to the cover and meets the remaining demand (=1) of $v(\{3,6\})$ and the demands of the inner vertices of the connectivity gadget and of the vertex $v(9^1)$.

For $v(\{4,5\})$, add the following edges to the cover.
\begin{itemize}
\item $(v(\{4,5\}), 1)$
\item $(2, \star)$
\item $(7, \star)$
\item $(8, \star)$
\item $(6, \star)$
\item $(3, 4)$
\item $(5, v(9^2))$
\end{itemize}
This adds 7 weight-1 edges to the cover and meets the remaining demand (=1) of $v(\{4,5\})$ and the demands of the inner vertices of the connectivity gadget and of the vertex $v(9^2)$.

The total weight of the cover described above is $22+4n^4$.

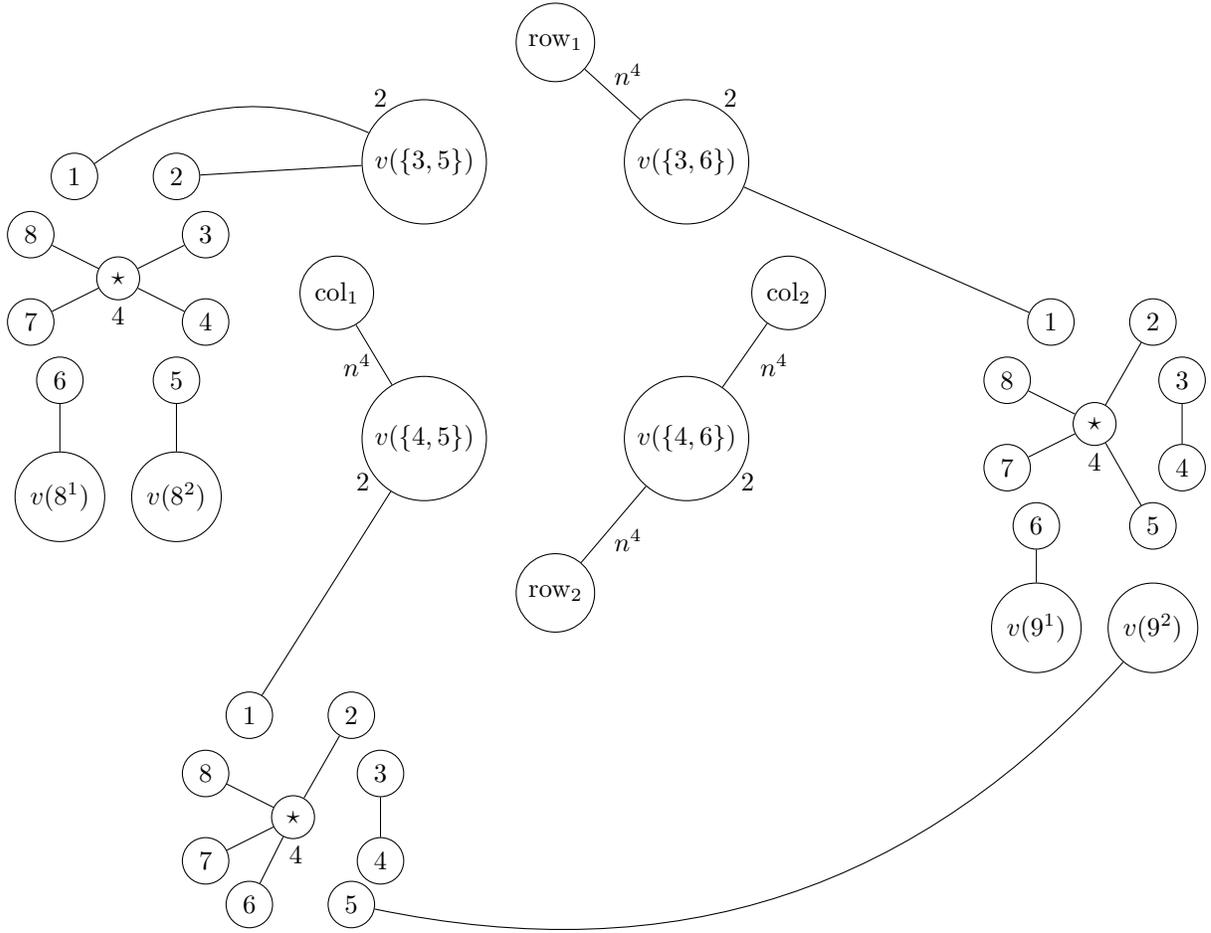
\begin{figure}[h]
\begin{center}
\begin{tikzpicture}[scale=0.775]
\tikzstyle{new style 0}=[fill=white, draw=black, shape=circle, draw=black]
\tikzstyle{none}=[]

\tikzstyle{new edge style 0}=[-, dashed]

\pgfdeclarelayer{edgelayer}
\pgfdeclarelayer{nodelayer}
\pgfsetlayers{edgelayer,nodelayer,main}
	\begin{pgfonlayer}{nodelayer}
		\node [style=new style 0] (0) at (10, 4.85) {row$_2$};
		\node [style=new style 0] (1) at (10, 14.3) {row$_1$};
		\node [style=new style 0] (2) at (14, 10) {col$_2$};
		\node [style=new style 0] (3) at (6.25, 10) {col$_1$};
		\node [style=new style 0] (4) at (7.75, 7.5) {$v(\{4,5\})$};
		\node [style=new style 0] (5) at (12.25, 7.5) {$v(\{4,6\})$};
		\node [style=new style 0] (6) at (7.75, 12.25) {$v(\{3,5\})$};
		\node [style=new style 0] (7) at (12.25, 12.25) {$v(\{3,6\})$};
		\node [style=new style 0] (8) at (1, 11) {$8$};
		\node [style=new style 0] (9) at (4, 11) {$3$};
		\node [style=new style 0] (10) at (2.5, 10.25) {$\star$};
		\node [style=new style 0] (11) at (1, 9.5) {$7$};
		\node [style=new style 0] (12) at (4, 9.5) {$4$};
		\node [style=new style 0] (13) at (1.75, 12) {$1$};
		\node [style=new style 0] (14) at (3.5, 12) {$2$};
		\node [style=new style 0] (15) at (1.5, 8.5) {$6$};
		\node [style=new style 0] (16) at (3.5, 8.5) {$5$};
		\node [style=none] (18) at (2.5, 9.6) {$4$};
		\node [style=new style 0] (19) at (17.75, 8.5) {$8$};
		\node [style=new style 0] (20) at (20.75, 8.5) {$3$};
		\node [style=new style 0] (21) at (19.25, 7.75) {$\star$};
		\node [style=new style 0] (22) at (17.75, 7) {$7$};
		\node [style=new style 0] (23) at (20.75, 7) {$4$};
		\node [style=new style 0] (24) at (18.5, 9.5) {$1$};
		\node [style=new style 0] (25) at (20.25, 9.5) {$2$};
		\node [style=new style 0] (26) at (18.25, 6) {$6$};
		\node [style=new style 0] (27) at (20.25, 6) {$5$};
		\node [style=none] (28) at (19.25, 7.1) {$4$};
		\node [style=new style 0] (29) at (4, 1.75) {$8$};
		\node [style=new style 0] (30) at (7, 1.75) {$3$};
		\node [style=new style 0] (31) at (5.5, 1) {$\star$};
		\node [style=new style 0] (32) at (4, 0.25) {$7$};
		\node [style=new style 0] (33) at (7, 0.25) {$4$};
		\node [style=new style 0] (34) at (4.75, 2.75) {$1$};
		\node [style=new style 0] (35) at (6.5, 2.75) {$2$};
		\node [style=new style 0] (36) at (4.75, -0.5) {$6$};
		\node [style=new style 0] (37) at (6.5, -0.5) {$5$};
		\node [style=none] (38) at (5.55, 0.35) {$4$};
		\node [style=new style 0] (39) at (18.25, 4.25) {$v(9^1)$};
		\node [style=new style 0] (40) at (20.25, 4.25) {$v(9^2)$};
		\node [style=new style 0] (41) at (1.5, 6.5) {$v(8^1)$};
		\node [style=new style 0] (42) at (3.5, 6.5) {$v(8^2)$};
		\node [style=none] (43) at (13.0, 13.35) {$2$};
		\node [style=none] (44) at (7.0, 13.35) {$2$};
		\node [style=none] (45) at (13.3, 6.75) {$2$};
		\node [style=none] (46) at (6.7, 6.75) {$2$};
		\node [style=none] (47) at (6.6, 8.75) {$n^4$};
		\node [style=none] (48) at (13.75, 8.75) {$n^4$};
		\node [style=none] (50) at (11.25, 13.75) {$n^4$};
		\node [style=none] (53) at (11.25, 5.75) {$n^4$};
	\end{pgfonlayer}
	\begin{pgfonlayer}{edgelayer}
		\draw (2) to (5);
		\draw (5) to (0);
		\draw (4) to (3);
		\draw (7) to (1);
		\draw (10) to (9);
		\draw (10) to (12);
		\draw (10) to (11);
		\draw (10) to (8);
		\draw (25) to (21);
		\draw (21) to (27);
		\draw (21) to (22);
		\draw (21) to (19);
		\draw (20) to (23);
		\draw (35) to (31);
		\draw (31) to (36);
		\draw (31) to (32);
		\draw (31) to (29);
		\draw (30) to (33);
		\draw (4) to (34);
		\draw (7) to (24);
		\draw (39) to (26);
		\draw [bend left] (40) to (37);
		\draw [bend left] (13) to (6);
		\draw (14) to (6);
		\draw (15) to (41);
		\draw (16) to (42);
	\end{pgfonlayer}
\end{tikzpicture}
\end{center}
\caption{$b$-Edge Cover of weight $22+4n^4$ for the graph constructed using the reduction in the proof of Theorem 2 from Bredereck and Talmon~\cite{bre-tal:j:edge-cover} given the negative NMTS instance $A = \{3,4\}$, $B = \{5,6\}$, $C=\{8,9\}$, and $n = 2$. Vertices and edges have $b$-values and weights of 1 except where labeled explicitly.}
\label{fig:counterexample-cover}
\end{figure}

\end{document}